\documentclass[11pt]{article}

\usepackage[margin=1in]{geometry}
\usepackage{color}
\usepackage{amssymb}
\usepackage{amsmath}
\usepackage{amsthm}
\usepackage{url}
\usepackage{graphicx}

\newtheorem{theorem}{Theorem}
\newtheorem{definition}[theorem]{Definition}

\newtheorem{lemma}[theorem]{Lemma}

\newtheorem{problem}[theorem]{Problem}

\begin{document}

\title{Discrete uncertainty principles and sparse signal processing}

\author{Afonso S.\ Bandeira\footnote{Department of Mathematics, Courant Institute of Mathematical Sciences, New York University, New York, NY} \qquad Megan E.\ Lewis\footnote{Detachment~5, Air Force Operational Test and Evaluation Center, Edwards AFB, CA} \qquad Dustin G.\ Mixon\footnote{Department of Mathematics and Statistics, Air Force Institute of Technology, Wright-Patterson AFB, OH}}

\maketitle

\begin{abstract}
We develop new discrete uncertainty principles in terms of numerical sparsity, which is a continuous proxy for the $0$-norm.
Unlike traditional sparsity, the continuity of numerical sparsity naturally accommodates functions which are nearly sparse.
After studying these principles and the functions that achieve exact or near equality in them, we identify certain consequences in a number of sparse signal processing applications.
\end{abstract}

\section{Introduction}

Uncertainty principles have maintained a significant role in both science and engineering for most of the past century.
In 1927, the concept was introduced by Werner Heisenberg in the context of quantum mechanics~\cite{Heisenberg:27}, in which a particle's position and momentum are represented by wavefunctions $f,g\in L^2(\mathbb{R})$, and $g$ happens to be the Fourier transform of $f$.
Measuring the position or momentum of a particle amounts to drawing a random variable whose probability density function is a normalized version of $|f|^2$ or $|g|^2$, respectively.
Heisenberg's uncertainty principle postulates a fundamental limit on the precision with which one can measure both position and momentum; in particular, the variance of the position measurement is small only if the momentum measurement exhibits large variance.
From a mathematical perspective, this physical principle can be viewed as an instance of a much broader meta-theorem in harmonic analysis:

\begin{quote}
A nonzero function and its Fourier transform cannot be simultaneously localized. 
\end{quote}

Heisenberg's uncertainty principle provides a lower bound on the product of the variances of the probability density functions corresponding to $f$ and $\hat{f}$.
In the time since, various methods have emerged for quantifying localization.
For example, instead of variance, one might consider entropy~\cite{Beckner:75}, the size of the density's support~\cite{AmreinB:77}, or how rapidly it decays~\cite{Hardy:33}.
Furthermore, the tradeoff in localization need not be represented by a product---as we will see, it is sometimes more telling to consider a sum.

Beyond physics, the impossibility of simultaneous localization has had significant consequences in signal processing.
For example, when working with the short-time Fourier transform, one is forced to choose between temporal and frequency resolution.
More recently, the emergence of digital signal processing has prompted the investigation of uncertainty principles underlying the discrete Fourier transform, notably by Donoho and Stark~\cite{DonohoS:89}, Tao~\cite{Tao:05}, and Tropp~\cite{Tropp:08}.
Associated with this line of work is the uniform uncertainty principle of Cand\`{e}s and Tao~\cite{CandesT:06}, which played a key role in the development of compressed sensing.
The present paper continues this investigation of discrete uncertainty principles with an eye on applications in sparse signal processing.

\subsection{Background and overview}

For any finite abelian group $G$, let $\ell(G)$ denote the set of functions $x\colon G\rightarrow\mathbb{C}$, and $\widehat{G}\subseteq\ell(G)$ the group of characters over $G$.
Then taking inner products with these characters and normalizing leads to the (unitary) Fourier transform $F\colon\ell(G)\rightarrow\ell(\widehat{G})$, namely
\[
(Fx)[\chi]
:=\frac{1}{\sqrt{|G|}}\sum_{g\in G}x[g]\overline{\chi[g]}
\qquad
\forall \chi\in \widehat{G}.
\]
The reader who is unfamiliar with Fourier analysis over finite abelian groups is invited to learn more in \cite{Terras:99}.
In the case where $G=\mathbb{Z}/n\mathbb{Z}$ (which we denote by $\mathbb{Z}_n$ in the sequel), the above definition coincides with the familiar discrete Fourier transform after one identifies characters with their frequencies.
The following theorem provides two uncertainty principles in terms of the so-called $0$-norm $\|\cdot\|_0$, defined to be number of nonzero entries in the argument.

\begin{theorem}[Theorem~1 in~\cite{DonohoS:89}, Theorem~1.1 in~\cite{Tao:05}]
\label{thm.tao up}
Let $G$ be a finite abelian group, and let $F\colon\ell(G)\rightarrow\ell(\widehat{G})$ denote the corresponding Fourier transform.
Then
\begin{equation}
\label{eq.multiplicative tao up}
\|x\|_0\|Fx\|_0\geq|G|\qquad\forall x\in\ell(G)\setminus\{0\}.
\end{equation}
Furthermore, if $|G|$ is prime, then
\begin{equation}
\label{eq.additive tao up}
\|x\|_0+\|Fx\|_0\geq|G|+1\qquad\forall x\in\ell(G)\setminus\{0\}.
\end{equation}
\end{theorem}

\begin{proof}[Proof sketch]
For \eqref{eq.multiplicative tao up}, apply the fact that the $\ell_1/\ell_\infty$-induced norm of $F$ is given by $\|F\|_{1\rightarrow\infty}=1/\sqrt{|G|}$, along with Cauchy--Schwarz and Parseval's identity: 
\[
\|Fx\|_{\infty}
\leq\frac{1}{\sqrt{|G|}}\|x\|_1\\
\leq\sqrt{\frac{\|x\|_0}{|G|}}\|x\|_2
=\sqrt{\frac{\|x\|_0}{|G|}}\|Fx\|_2
\leq\sqrt{\frac{\|x\|_0\|Fx\|_0}{|G|}}\|Fx\|_{\infty},
\]
where the last step bounds a sum in terms of its largest summand.
Rearranging gives the result.

For \eqref{eq.additive tao up}, suppose otherwise that there exists $x\neq0$ which violates the claimed inequality.
Denote $\mathcal{J}=\operatorname{supp}(x)$ and pick some $\mathcal{I}\subseteq\widehat{G}\setminus\operatorname{supp}(Fx)$ with $|\mathcal{I}|=|\mathcal{J}|$.
Then $0=(Fx)_\mathcal{I}=F_{\mathcal{I}\mathcal{J}}x_\mathcal{J}$.
Since the submatrix $F_{\mathcal{I}\mathcal{J}}$ is necessarily invertible by a theorem of Chebotar\"{e}v~\cite{StevenhagenL:96}, we conclude that $x_\mathcal{J}=0$, a contradiction.
\end{proof}

We note that the additive uncertainty principle above is much stronger than its multiplicative counterpart.
Indeed, with the help of the arithmetic mean--geometric mean inequality, \eqref{eq.multiplicative tao up} immediately implies
\begin{equation}
\label{eq.amgm}
\|x\|_0+\|Fx\|_0
\geq2\sqrt{\|x\|_0\|Fx\|_0}
\geq2\sqrt{|G|}
\qquad
\forall x\in\ell(G),
\end{equation}
which is sharp when $G=\mathbb{Z}_n$ and $n$ is a perfect square (simply take $x$ to be a Dirac comb, specifically, the indicator function $1_K$ of the subgroup $K$ of size $\sqrt{n}$).
More generally, if $n$ is not prime, then $n=ab$ with integers $a,b\in[2,n/2]$, and so $a+b\leq n/2+2<n+1$; as such, taking $x$ to be an indicator function of the subgroup of size $a$ (whose Fourier transform necessarily has $0$-norm $b$) will violate \eqref{eq.additive tao up}.
Overall, the hypothesis that $|G|$ is prime cannot be weakened.
Still, something can be said if one slightly strengthens the hypothesis on $x$.
For example, Theorem~A in~\cite{Tropp:10} gives that for every $S\subseteq G$,
\[
\|x\|_0+\|Fx\|_0
>\sqrt{|G|\|x\|_0}
\]
for almost every $x\in\ell(G)$ supported on $S$.
This suggests that extreme functions like the Dirac comb are atypical, i.e., \eqref{eq.amgm} is ``barely sharp.''

One could analogously argue that, in some sense, \eqref{eq.additive tao up} is ``barely true'' when $|G|$ is prime.
For an illustration, Figure~\ref{figure.discrete gaussian} depicts a discrete version of the Gaussian function, which is constructed by first periodizing the function $f(t)=e^{-n\pi t^2}$ over the real line in order to have unit period, and then sampling this periodized function at multiples of $1/n$.
As we verify in subsection~\ref{subsection.near equality}, the resulting function $x\in\ell(\mathbb{Z}_n)$ satisfies $Fx=x$, analogous to the fact that a Gaussian function in $L^2(\mathbb{R})$ with the proper width is fixed by the Fourier transform.
Given its resemblance to the fast-decaying Gaussian function over $\mathbb{R}$, it comes as no surprise that many entries of this function are nearly zero.
In the depicted case where $n=211$ (which is prime), only $99$ entries of this function manage to be larger than machine precision, and so from a numerical perspective, this function appears to contradict Theorem~\ref{thm.tao up}: $99+99=198<212=211+1$.

\begin{figure}[t]
\begin{center}
\includegraphics[width=0.95\textwidth]{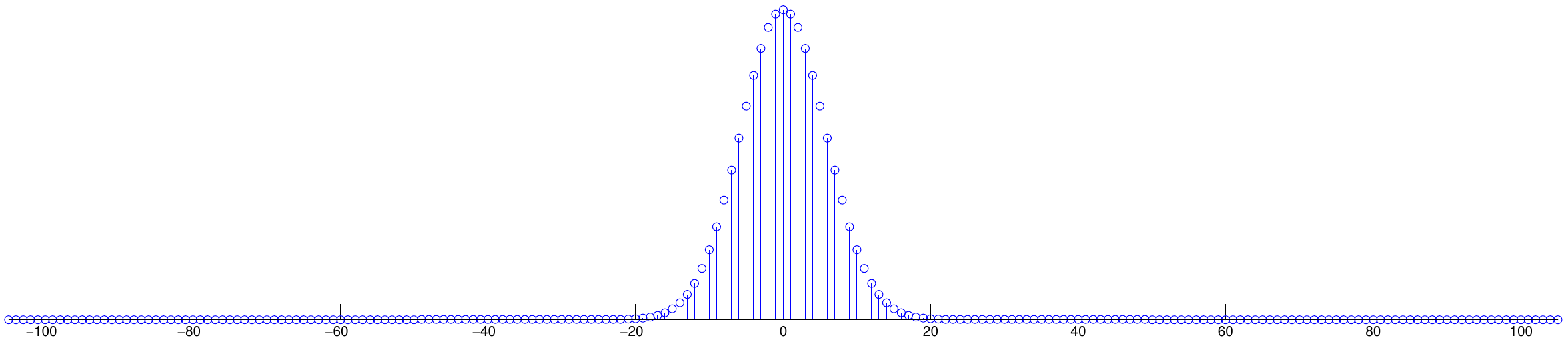}
\end{center}
\caption{
\label{figure.discrete gaussian}
{\footnotesize 
Discrete Gaussian function, obtained by periodizing the function $f(t)=e^{-n\pi t^2}$ with period $1$ before sampling at multiples of $1/n$.
The resulting function in $\mathbb{Z}_n$ is fixed by the $n\times n$ discrete Fourier transform.
In this figure, we take $n=211$, and only $99$ entries are larger than machine precision (i.e., $2.22\times10^{-16}$).
As such, an unsuspecting signal processor might think $\|x\|_0$ and $\|Fx\|_0$ are both $99$ instead of $211$.
Since $211$ is prime and $99+99=198<212=211+1$, this illustrates a lack of numerical robustness in the additive uncertainty principle of Theorem~\ref{thm.tao up}.
By contrast, our main result (Theorem~\ref{thm:main}) provides a robust alternative in terms of numerical sparsity, though the result is not valid for the discrete Fourier transform, but rather a random unitary matrix.
}
\normalsize}
\end{figure}

To help resolve this discrepancy, we consider a numerical version of traditional sparsity which is aptly named \textbf{numerical sparsity}:
\[
\operatorname{ns}(x)
:=\frac{\|x\|_1^2}{\|x\|_2^2}
\qquad
\forall x\in\mathbb{C}^n\setminus\{0\}.
\]
See Figure~\ref{figure.numerical sparsity} for an illustration.
This ratio appeared as early as 1978 in the context of geophysics~\cite{Gray:78}.
More recently, it has been used as a proxy for sparsity in various signal processing applications~\cite{HurleyR:09,Lopez:arxiv12,DemanetH:14,RepettiPDCP:14,Studer:15}.
The numerical rank of a matrix is analogously defined as the square ratio of the nuclear and Frobenius norms, and has been used, for example, in Alon's work on extremal combinatorics~\cite{Alon:03}.
We note that numerical sparsity is invariant under nonzero scaling, much like traditional sparsity.
In addition, one bounds the other:
\begin{equation}
\label{eq.ns vs 0norm}
\operatorname{ns}(x)
\leq\|x\|_0.
\end{equation}
To see this, apply Cauchy--Schwarz to get
\[
\|x\|_1
=\langle |x|,1_{\operatorname{supp}(x)}\rangle
\leq \|x\|_2\|1_{\operatorname{supp}(x)}\|_2
=\|x\|_2\sqrt{\|x\|_0},
\]
where $|x|$ denotes the entrywise absolute value of $x$.
Rearranging then gives \eqref{eq.ns vs 0norm}.
For this paper, the most useful feature of numerical sparsity is its continuity, as this will prevent near-counterexamples like the one depicted in Figure~\ref{figure.discrete gaussian}.
What follows is our main result, which leverages numerical sparsity to provide uncertainty principles that are analogous to those in Theorem~\ref{thm.tao up}:

\begin{figure}[t]
\begin{center}
\includegraphics[width=0.45\textwidth]{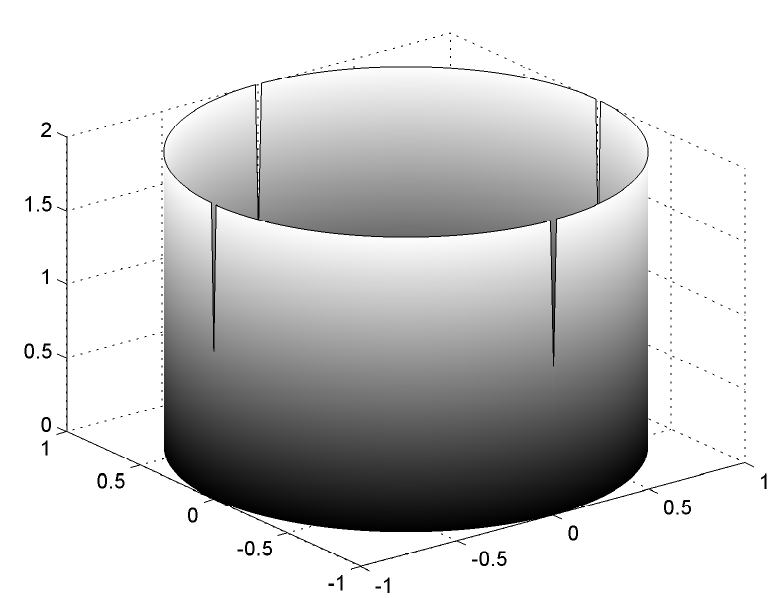}
\includegraphics[width=0.45\textwidth]{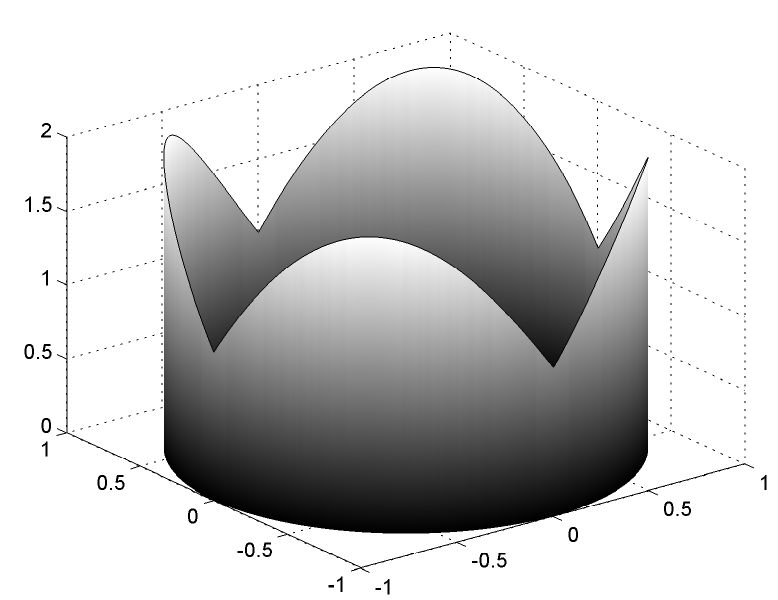}
\end{center}
\caption{
\label{figure.numerical sparsity}
{\footnotesize 
Traditional sparsity $\|x\|_0$ (left) and numerical sparsity $\operatorname{ns}(x)$ (right) for all $x$ in the unit circle in $\mathbb{R}^2$.
This illustrates how numerical sparsity is a continuous analog of traditional sparsity; we leverage this feature to provide robust alternatives to the uncertainty principles of Theorem~\ref{thm.tao up}.
In this case, one may verify that $\operatorname{ns}(x)\leq\|x\|_0$ by visual inspection.
}
\normalsize}
\end{figure}

\begin{theorem}[Main result\footnote{Recall that $f(n)=O(g(n))$ if there exists $C,n_0>0$ such that $f(n)\leq Cg(n)$ for all $n>n_0$. We write $f(n)=O_\delta(g(n))$ if the constant $C$ is a function of $\delta$. Also, $f(n)=\Omega(g(n))$ if $g(n)=O(f(n))$, and $f(n)=o(g(n))$ if $f(n)/g(n)\rightarrow 0$ as $n\rightarrow\infty$.}]
\label{thm:main}
Let $U$ be an $n\times n$ unitary matrix. Then
\begin{equation}
\label{eq.main 1}
\operatorname{ns}(x)\operatorname{ns}(Ux)\geq\frac{1}{\|U\|_{1\rightarrow\infty}^2}\qquad\forall x\in\mathbb{C}^n\setminus\{0\},
\end{equation}
where $\|\cdot\|_{1\rightarrow\infty}$ denotes the induced matrix norm.
Furthermore, there exists a universal constant $c>0$ such that if $U$ is drawn uniformly from the unitary group $\operatorname{U}(n)$, then with probability $1-e^{-\Omega(n)}$,
\begin{equation}
\label{eq.main 2}
\operatorname{ns}(x)+\operatorname{ns}(Ux)\geq (c-o(1))n\qquad\forall x\in\mathbb{C}^n\setminus\{0\}.
\end{equation}
\end{theorem}

Perhaps the most glaring difference between Theorems~\ref{thm.tao up} and~\ref{thm:main} is our replacement of the Fourier transform with an arbitrary unitary matrix.
Such generalizations have appeared in the quantum physics literature (for example, see~\cite{Kraus:87}), as well as in the sparse signal processing literature~\cite{DonohoH:01,DonohoE:03,GribonvalN:07,Tropp:08,Tropp:10,Studer:15}.
Our multiplicative uncertainty principle still applies when $U=F$, in which case $\|U\|_{1\rightarrow\infty}=1/\sqrt{n}$.
Considering \eqref{eq.ns vs 0norm}, the uncertainty principle in this case immediately implies the analogous principle in Theorem~\ref{thm.tao up}.
Furthermore, the proof is rather straightforward:
Apply H\"{o}lder's inequality to get
\begin{equation}
\label{eq.proof of main result}
\operatorname{ns}(x)\operatorname{ns}(Ux)=\frac{\|x\|_1^2}{\|x\|_2^2}\cdot\frac{\|Ux\|_1^2}{\|Ux\|_2^2}
\geq\frac{\|x\|_1^2}{\|x\|_2^2}\cdot\frac{\|Ux\|_2^2}{\|Ux\|_{\infty}^2}
=\frac{\|x\|_1^2}{\|Ux\|_{\infty}^2}
\geq\frac{1}{\|U\|_{1\rightarrow\infty}^2}.
\end{equation}

By contrast, the proof of our additive uncertainty principle is not straightforward, and it does not hold if we replace $U$ with $F$.
Indeed, as we show in subsection~\ref{subsection.near equality}, the discrete Gaussian function depicted in Figure~\ref{figure.discrete gaussian} has numerical sparsity $O(\sqrt{n})$, thereby violating \eqref{eq.main 2}; recall that the same function is a near-counterexample of the analogous principle in Theorem~\ref{thm.tao up}.
Interestingly, our uncertainty principle establishes that the Fourier transform is rare in that the vast majority of unitary matrices offer much more uncertainty in the worst case.
This naturally leads to the following question:

\begin{problem}
For each $n$, what is the largest $c=c(n)$ for which there exists a unitary matrix $U$ that satisfies $\operatorname{ns}(x)+\operatorname{ns}(Ux)\geq cn$ for every $x\in\mathbb{C}^n\setminus\{0\}$?
\end{problem}

Letting $x=e_1$ gives $\operatorname{ns}(x)+\operatorname{ns}(Ux)\leq1+\|Ux\|_0\leq n+1$, and so $c(n)\leq 1+o(1)$; a bit more work produces a strict inequality $c(n)<1+1/n$ for $n\geq4$.
Also, our proof of the uncertainty principle implies $\liminf_{n\rightarrow\infty}c(n)\geq1/540000$.

\subsection{Outline}

The primary focus of this paper is Theorem~\ref{thm:main}.
Having already proved the multiplicative uncertainty principle in \eqref{eq.proof of main result}, it remains to prove the additive counterpart, which we do in the following section.
Next, Section~3 considers functions which achieve either exact or near equality in \eqref{eq.main 1} when $U$ is the discrete Fourier transform.
Surprisingly, exact equality occurs in \eqref{eq.main 1} precisely when it occurs in \eqref{eq.multiplicative tao up}.
We also show that the discrete Gaussian depicted in Figure~\ref{figure.discrete gaussian} achieves near equality in \eqref{eq.main 1}.
We conclude in Section~4 by studying a few applications, specifically, sparse signal demixing, compressed sensing with partial Fourier operators, and the fast detection of sparse signals.

\section{Proof of additive uncertainty principle}

In this section, we prove the additive uncertainty principle in Theorem~\ref{thm:main}.
The following provides a more explicit statement of the principle we prove:

\begin{theorem}
\label{thm.additive up}
Draw $U$ uniformly from the unitary group $\operatorname{U}(n)$.
Then with probability $\geq1-8e^{-n/4096}$,
\[
\operatorname{ns}(x)+\operatorname{ns}(Ux)\geq\frac{1}{9}\left\lfloor\frac{n}{60000}\right\rfloor\qquad\forall x\in\mathbb{C}^n\setminus\{0\}.
\]
\end{theorem}

For the record, we did not attempt to optimize the constants.
Our proof of this theorem makes use of several ideas from the compressed sensing literature:

\begin{definition}
Take any $m\times n$ matrix $\Phi=[\varphi_1\cdots\varphi_n]$.
\begin{itemize}
\item[(a)]
We say $\Phi$ exhibits $(k,\theta)$-\textbf{restricted orthogonality} if
\[
|\langle \Phi x,\Phi y\rangle|\leq\theta\|x\|_2\|y\|_2
\]
for every $x,y\in\mathbb{C}^n$ with $\|x\|_0,\|y\|_0\leq k$ and disjoint support.
\item[(b)]
We say $\Phi$ satisfies the $(k,\delta)$-\textbf{restricted isometry property} if
\[
(1-\delta)\|x\|_2^2\leq\|\Phi x\|_2^2\leq(1+\delta)\|x\|_2^2
\]
for every $x\in\mathbb{C}^n$ with $\|x\|_0\leq k$.
\item[(c)]
We say $\Phi$ satisfies the $(k,c)$-\textbf{width property} if
\[
\|x\|_2\leq\frac{c}{\sqrt{k}}\|x\|_1
\]
for every $x$ in the nullspace of $\Phi$.
\end{itemize}
\end{definition}

The restricted isometry property is a now-standard sufficient condition for uniformly stable and robust reconstruction from compressed sensing measurements (for example, see~\cite{Candes:08}).
As the following statement reveals, restricted orthogonality implies the restricted isometry property:

\begin{lemma}[Lemma~11 in~\cite{BandeiraFMW:13}]
\label{lemma.ro to rip}
If a matrix satisfies $(k,\theta)$-restricted orthogonality and its columns have unit norm, then it also satisfies the $(k,\delta)$-restricted isometry property with $\delta=2\theta$.
\end{lemma}

To prove Theorem~\ref{thm.additive up}, we will actually make use of the width property, which was introduced by Kashin and Temlyakov~\cite{KashinT:07} to characterize uniformly stable $\ell_1$ reconstruction for compressed sensing.
Luckily, the restricted isometry property implies the width property:

\begin{lemma}[Theorem~11 in~\cite{CahillM:arxiv14}, cf.~\cite{KashinT:07}]
\label{lemma.rip to wp}
If a matrix satisfies the $(k,\delta)$-restricted isometry property for some positive integer $k$ and $\delta<1/3$, then it also satisfies the $(k,3)$-width property.
\end{lemma}

What follows is a stepping-stone result that we will use to prove Theorem~\ref{thm.additive up}, but it is also of independent interest:

\begin{theorem}
\label{thm.iu rip}
Draw $U$ uniformly from the unitary group $\operatorname{U}(n)$.
Then $[I~U]$ satisfies the $(k,\delta)$-restricted isometry property with probability $\geq1-8e^{-\delta^2n/256}$ provided $\delta<1$ and
\begin{equation}
\label{eq.choice of k for IU rip}
n\geq\frac{256}{\delta^2}k\log\bigg(\frac{en}{k}\bigg).
\end{equation}
\end{theorem}

This is perhaps not surprising, considering various choices of structured random matrices are known to form restricted isometries with high probability~\cite{CandesT:06,RudelsonV:08,Rauhut:10,KrahmerMR:14,NelsonPW:14,Bourgain:14,BandeiraFMM:arxiv14}.
To prove Theorem~\ref{thm.iu rip}, we show that the structured matrix enjoys restricted orthogonality with high probability, and then appeal to Lemma~\ref{lemma.ro to rip}.
Before proving this result, we first motivate it by proving the desired uncertainty principle:

\begin{proof}[Proof of Theorem~\ref{thm.additive up}]
Take $k=\lfloor n/60000\rfloor$ and $\delta=1/4$.
We will show $\operatorname{ns}(x)+\operatorname{ns}(Ux)\geq k/9$ for every nonzero $x\in\mathbb{C}^n$.
If $k=0$, the result is immediate, and so $n\geq60000$ without loss of generality.
In this regime, we have $k\in[n/120000,n/60000]$, and so
\[
\frac{256}{\delta^2}k\log\bigg(\frac{en}{k}\bigg)
\leq4096\log(120000e)\cdot k
\leq60000k
\leq n.
\]
Theorem~\ref{thm.iu rip} and Lemma~\ref{lemma.rip to wp} then give that $[I~U]$ satisfies the $(k,3)$-width property with probability $\geq1-8e^{-n/4096}$.
Observe that $z=[Ux;-x]$ resides in the nullspace of $[I~U]$ regardless of $x\in\mathbb{C}^n$.
In the case where $x$ (and therefore $z$) is nonzero, the width property and the arithmetic mean--geometric mean inequality together give
\[
\frac{k}{9}
\leq\frac{\|z\|_1^2}{\|z\|_2^2}
=\frac{(\|x\|_1+\|Ux\|_1)^2}{\|x\|_2^2+\|Ux\|_2^2}
=\frac{\|x\|_1^2+2\|x\|_1\|Ux\|_1+\|Ux\|_1^2}{2\|x\|_2^2}
\leq\operatorname{ns}(x)+\operatorname{ns}(Ux).\qedhere
\]
\end{proof}

\begin{proof}[Proof of Theorem~\ref{thm.iu rip}]
Take $[I~U]=[\varphi_1\cdots\varphi_{2n}]$, and let $k$ be the largest integer satisfying $\eqref{eq.choice of k for IU rip}$.
We will demonstrate that $[I~U]$ satisfies the $(k,\delta)$-restricted isometry property, which will then imply the $(k',\delta)$-restricted isometry property for all $k'<k+1$, and therefore all $k$ satisfying \eqref{eq.choice of k for IU rip}.
To this end, define the random quantities
\[
\theta^\star(U):=\max_{\substack{x,y\in\mathbb{C}^{2n}\\\|x\|_0,\|y\|_0\leq k\\\operatorname{supp}(x)\cap\operatorname{supp}(y)=\emptyset}}\frac{|\langle \Phi x,\Phi y\rangle|}{\|x\|_2\|y\|_2},
\qquad
\theta(U):=\max_{\substack{x,y\in\mathbb{C}^{2n}\\\|x\|_0,\|y\|_0\leq k\\\operatorname{supp}(x)\subseteq[n]\\\operatorname{supp}(y)\subseteq[n]^c}}\frac{|\langle\Phi x,\Phi y\rangle|}{\|x\|_2\|y\|_2}.
\]
We first claim that $\theta^\star(U)\leq\theta(U)$.
To see this, for any $x,y$ satisfying the constraints in $\theta^\star(U)$, decompose $x=x_1+x_2$ so that $x_1$ and $x_2$ are supported in $[n]$ and $[n]^c$, respectively, and similarly $y=y_1+y_2$.
For notational convenience, let $S$ denote the set of all $4$-tuples $(a,b,c,d)$ of $k$-sparse vectors in $\mathbb{C}^{2n}$ such that $a$ and $b$ are disjointly supported in $[n]$, while $c$ and $d$ are disjointly supported in $[n]^c$.
Then $(x_1,y_1,x_2,y_2)\in S$.
Since $\operatorname{supp}(x)$ and $\operatorname{supp}(y)$ are disjoint, and since $I$ and $U$ each have orthogonal columns, we have
\[
\langle \Phi x,\Phi y\rangle
=\langle \Phi x_1,\Phi y_2\rangle+\langle \Phi x_2,\Phi y_1\rangle.
\]
As such, the triangle inequality gives
\begin{align*}
\theta^\star(U)
&=\max_{\substack{x,y\in\mathbb{C}^{2n}\\\|x\|_0,\|y\|_0\leq k\\\operatorname{supp}(x)\cap\operatorname{supp}(y)=\emptyset}}\frac{|\langle \Phi x_1,\Phi y_2\rangle+\langle \Phi x_2,\Phi y_1\rangle|}{\|x\|_2\|y\|_2}\\
&\leq\max_{(x_1,y_1,x_2,y_2)\in S}\frac{|\langle \Phi x_1,\Phi y_2\rangle|+|\langle\Phi x_2,\Phi y_1\rangle|}{\sqrt{\|x_1\|_2^2+\|x_2\|_2^2}\sqrt{\|y_1\|_2^2+\|y_2\|_2^2}}\\
&\leq\bigg(\max_{(x_1,y_1,x_2,y_2)\in S}\frac{\|x_1\|_2\|y_2\|_2+\|x_2\|_2\|y_1\|_2}{\sqrt{\|x_1\|_2^2+\|x_2\|_2^2}\sqrt{\|y_1\|_2^2+\|y_2\|_2^2}}\bigg)\theta(U)\\
&\leq\theta(U),
\end{align*}
where the last step follows from squaring and applying the arithmetic mean--geometric mean inequality:
\[
\bigg(\frac{\sqrt{ad}+\sqrt{bc}}{\sqrt{(a+b)(c+d)}}\bigg)^2
=\frac{ad+bc+2\sqrt{acbd}}{(a+b)(c+d)}
\leq\frac{ad+bc+(ac+bd)}{(a+b)(c+d)}
=1.
\]
At this point, we seek to bound the probability that $\theta(U)$ is large.
First, we observe an equivalent expression:
\[
\theta(U)=\max_{\substack{x,y\in\mathbb{C}^n\\\|x\|_2=\|y\|_2=1\\\|x\|_0,\|y\|_0\leq k}}|\langle x,Uy\rangle|.
\]
To estimate the desired probability, we will pass to an $\epsilon$-net $\mathcal{N}_\epsilon$ of $k$-sparse vectors with unit $2$-norm.
A standard volume-comparison argument gives that the unit sphere in $\mathbb{R}^{m}$ enjoys an $\epsilon$-net of size $\leq(1+2/\epsilon)^m$ (see Lemma~5.2 in~\cite{Vershynin:12}).
As such, for each choice of $k$ coordinates, we can cover the corresponding copy of the unit sphere in $\mathbb{C}^k=\mathbb{R}^{2k}$ with $\leq(1+2/\epsilon)^{2k}$ points, and unioning these produces an $\epsilon$-net of size
\[
|\mathcal{N}_\epsilon|
\leq\binom{n}{k}\bigg(1+\frac{2}{\epsilon}\bigg)^{2k}.
\]
To apply this $\epsilon$-net, we note that $\|x-x'\|_2,\|y-y'\|_2\leq\epsilon$ and $\|x'\|_2=\|y'\|_2=1$ together imply
\begin{align*}
|\langle x,Uy\rangle|
&=|\langle x'+x-x',U(y'+y-y')\rangle|\\
&\leq |\langle x',Uy'\rangle|+\|x-x'\|_2+\|y-y'\|_2+\|x-x'\|_2\|y-y'\|_2\\
&\leq |\langle x',Uy'\rangle|+3\epsilon,
\end{align*}
where the last step assumes $\epsilon\leq1$.
As such, the union bound gives
\begin{align}
\nonumber
\operatorname{Pr}(\theta(U)>t)
&=\operatorname{Pr}\bigg(\exists x,y\in\mathbb{C}^n,~\|x\|_2=\|y\|_2=1,~\|x\|_0,\|y\|_0\leq k~\mbox{s.t.}~|\langle x,Uy\rangle|>t\bigg)\\
\nonumber
&\leq\operatorname{Pr}\bigg(\exists x,y\in\mathcal{N}_\epsilon~\mbox{s.t.}~|\langle x,Uy\rangle|>t-3\epsilon\bigg)\\
\nonumber
&\leq\sum_{x,y\in\mathcal{N}_\epsilon}\operatorname{Pr}\Big(|\langle x,Uy\rangle|>t-3\epsilon\Big)\\
\label{eq.unitary bound 1}
&=\binom{n}{k}^2\bigg(1+\frac{2}{\epsilon}\bigg)^{4k}\operatorname{Pr}\Big(|\langle e_1,Ue_1\rangle|>t-3\epsilon\Big),
\end{align}
where the last step uses the fact that the distribution of $U$ is invariant under left- and right-multiplication by any deterministic unitary matrix (e.g., unitary matrices that send $e_1$ to $x$ and $y$ to $e_1$, respectively).
It remains to prove tail bounds on $U_{11}:=\langle e_1,Ue_1\rangle$.
First, we apply the union bound to get
\begin{equation}
\label{eq.unitary bound 2}
\operatorname{Pr}(|U_{11}|>u)
\leq\operatorname{Pr}\bigg(|\operatorname{Re}(U_{11})|>\frac{u}{\sqrt{2}}\bigg)+\operatorname{Pr}\bigg(|\operatorname{Im}(U_{11})|>\frac{u}{\sqrt{2}}\bigg)
=4\operatorname{Pr}\bigg(\operatorname{Re}(U_{11})>\frac{u}{\sqrt{2}}\bigg),
\end{equation}
where the last step uses the fact that $\operatorname{Re}(U_{11})$ has even distribution.
Next, we observe that $\operatorname{Re}(U_{11})$ has the same distribution as $g/\sqrt{h}$, where $g$ has standard normal distribution and $h$ has chi-squared distribution with $2n$ degrees of freedom.
Indeed, this can be seen from one method of constructing the matrix $U$:
Start with an $n\times n$ matrix $G$ with iid $N(0,1)+iN(0,1)$ complex Gaussian entries and apply Gram--Schmidt to the columns; the first column of $U$ is then the first column of $G$ divided by its norm $\sqrt{h}$.
Let $s>0$ be arbitrary (to be selected later).
Then $g/\sqrt{h}>u/\sqrt{2}$ implies that either $g>\sqrt{s}u/\sqrt{2}$ or $h<s$.
As such, the union bound implies
\begin{equation}
\label{eq.unitary bound 3}
\operatorname{Pr}\bigg(\operatorname{Re}(U_{11})>\frac{u}{\sqrt{2}}\bigg)
\leq2\max\bigg\{\operatorname{Pr}\bigg(g>\sqrt{s}\frac{u}{\sqrt{2}}\bigg),\operatorname{Pr}(h<s)\bigg\}.
\end{equation}
For the first term, Proposition~7.5 in~\cite{FoucartR:13} gives
\begin{equation}
\label{eq.unitary bound 4}
\operatorname{Pr}\bigg(g>\sqrt{s}\frac{u}{\sqrt{2}}\bigg)
\leq e^{-su^2/4}.
\end{equation}
For the second term, Lemma~1 in~\cite{LaurentM:00} gives $\operatorname{Pr}(h<2n-\sqrt{8nx})\leq e^{-x}$ for any $x>0$.
Picking $x=(2n-s)^2/(8n)$ then gives
\begin{equation}
\label{eq.unitary bound 5}
\operatorname{Pr}(h<s)\leq e^{-(2n-s)^2/(8n)}.
\end{equation}
We use the estimate $\binom{n}{k}\leq(en/k)^k$ when combining \eqref{eq.unitary bound 1}--\eqref{eq.unitary bound 5} to get
\[
\log\Big(\operatorname{Pr}(\theta(U)>t)\Big)
\leq 2k\log\bigg(\frac{en}{k}\bigg)+4k\log\bigg(1+\frac{2}{\epsilon}\bigg)+\log8-\min\bigg\{\frac{s(t-3\epsilon)^2}{4},\frac{(2n-s)^2}{8n}\bigg\}.
\]
Notice $n/k\geq(256/\delta^2)\log(en/k)\geq256$ implies that taking $\epsilon=\sqrt{(k/n)\log(en/k)}$ gives
\[
\sqrt{\frac{en}{k}}-\frac{2}{\epsilon}
=\bigg(1-\frac{2}{\sqrt{e\log(n/k)}}\bigg)\sqrt{\frac{en}{k}}
\geq\bigg(1-\frac{2}{\sqrt{e\log(256)}}\bigg)\sqrt{256e}
\geq1,
\]
which can be rearranged to get
\[
\log\bigg(1+\frac{2}{\epsilon}\bigg)
\leq\frac{1}{2}\log\bigg(\frac{en}{k}\bigg).
\]
As such, we also pick $s=n$ and $t=\sqrt{(64k/n)\log(en/k)}$ to get
\[
\log\Big(\operatorname{Pr}(\theta(U)>t)\Big)
\leq 4k\log\bigg(\frac{en}{k}\bigg)+\log 8-\frac{25}{4}k\log\bigg(\frac{en}{k}\bigg)
\leq \log 8-2k\log\bigg(\frac{en}{k}\bigg).
\]
Since we chose $k$ to be the largest integer satisfying \eqref{eq.choice of k for IU rip}, we therefore have $\theta(U)\leq\sqrt{(64k/n)\log(n/k)}$ with probability $\geq1-8e^{-\delta^2n/256}$.
Lemma~\ref{lemma.ro to rip} then gives the result.
\end{proof}

\section{Low uncertainty with the discrete Fourier transform}

In this section, we study functions which achieve either exact or near equality in our multiplicative uncertainty principle \eqref{eq.main 2} in the case where the unitary matrix $U$ is the discrete Fourier transform.

\subsection{Exact equality in the multiplicative uncertainty principle}

We seek to understand when equality is achieved in \eqref{eq.main 2} in the special case of the discrete Fourier transform.
For reference, the analogous result for \eqref{eq.multiplicative tao up} is already known:

\begin{theorem}[Theorem~13 in~\cite{DonohoS:89}]
\label{thm.equality in old}
Suppose $x\in\ell(\mathbb{Z}_n)$ satisfies $\|x\|_0\|Fx\|_0=n$.
Then $x$ has the form $x=cT^aM^b1_K$, where $c\in\mathbb{C}$, $K$ is a subgroup of $\mathbb{Z}_n$, and $T,M\colon\ell(\mathbb{Z}_n)\rightarrow\ell(\mathbb{Z}_n)$ are translation and modulation operators defined by
\[
(Tx)[j]:=x[j-1],
\qquad
(Mx)[j]:=e^{2\pi ij/n}x[j]
\qquad
\forall j\in\mathbb{Z}_n.
\]
Here, $i$ denotes the imaginary unit $\sqrt{-1}$.
\end{theorem}

In words, equality is achieved in \eqref{eq.multiplicative tao up} by indicator functions of subgroups, namely, the so-called Dirac combs (as well as their scalar multiples, translations, modulations).
We seek an analogous characterization for our uncertainty principle \eqref{eq.main 2}.
Surprisingly, the characterization is identical:

\begin{theorem}\label{thm.equal}
Suppose $x\in\ell(\mathbb{Z}_n)$.
Then $\operatorname{ns}(x)\operatorname{ns}(Fx)=n$ if and only if $\|x\|_0\|Fx\|_0=n$.
\end{theorem}

\begin{proof}
($\Leftarrow$)
This follows directly from \eqref{eq.ns vs 0norm}, along with Theorems~\ref{thm.tao up} and~\ref{thm:main}.

($\Rightarrow$)
It suffices to show that $\operatorname{ns}(x)=\|x\|_0$ and $\operatorname{ns}(Fx)=\|Fx\|_0$.
Note that both $F$ and $F^{-1}$ are unitary operators and $\|F\|_{1\rightarrow\infty}^2=\|F^{-1}\|_{1\rightarrow\infty}^2=1/n$. 
By assumption, taking $y:=Fx$ then gives
\[
\operatorname{ns}(F^{-1}y)\operatorname{ns}(y)=\operatorname{ns}(x)\operatorname{ns}(Fx)=n.
\]
We will use the fact that $x$ and $y$ each achieve equality in the first part of Theorem~\ref{thm:main} with $U=F$ and $U=F^{-1}$, respectively.
Notice from the proof \eqref{eq.proof of main result} that equality occurs only if $x$ and $y$ satisfy equality in H\"{o}lder's inequality, that is,
\begin{equation}
\label{eq: first}
\|x\|_1\|x\|_{\infty}=\|x\|_2^2,\qquad \|y\|_1\|y\|_{\infty}=\|y\|_2^2.
\end{equation}
To achieve the first equality in \eqref{eq: first},
\[
\sum_{j\in\mathbb{Z}_n}|x[j]|^2
=\|x\|_2^2
=\|x\|_1\|x\|_{\infty}
=\sum_{j\in \mathbb{Z}_n}|x[j]|\max_{k\in\mathbb{Z}_n}|x[k]|.
\]
This implies that $|x[j]|=\max_k|x[k]|$ for every $j$ with $x[j]\neq0$.
Similarly, in order for the second equality in \eqref{eq: first} to hold, $|y[j]|=\max_k|y[k]|$ for every $j$ with $y[j]\neq0$.
As such, $|x|=a1_A$ and $|y|=b1_B$ for some $a,b>0$ and $A,B\subseteq\mathbb{Z}_n$.
Then
\[
\operatorname{ns}(x)
=\frac{\|x\|_1^2}{\|x\|_2^2}
=\frac{(a|A|)^2}{a^2|A|}
=|A|
=\|x\|_0,
\]
and similarly, $\operatorname{ns}(y)=\|y\|_0$.
\end{proof}

\subsection{Near equality in the multiplicative uncertainty principle}
\label{subsection.near equality}

Having established that equality in the new multiplicative uncertainty principle \eqref{eq.main 1} is equivalent to equality in the analogous principle \eqref{eq.multiplicative tao up}, we wish to separate these principles by focusing on near equality.
For example, in the case where $n$ is prime, $\mathbb{Z}_n$ has no nontrivial proper subgroups, and so by Theorem~\ref{thm.equality in old}, equality is only possible with identity basis elements and complex exponentials.
On the other hand, we expect the new principle to accommodate nearly sparse vectors, and so we appeal to the discrete Gaussian depicted in Figure~\ref{figure.discrete gaussian}:

\begin{theorem}
\label{thm.neareqmnup}
Define $x\in\ell(\mathbb{Z}_n)$ by
\begin{equation}
\label{eq.discrete gaussian}
x[j]
:=\sum_{j'\in\mathbb{Z}}e^{-n\pi(\frac{j}{n}+j')^2}
\qquad
\forall j\in\mathbb{Z}_n.
\end{equation}
Then $Fx=x$ and $\operatorname{ns}(x)\operatorname{ns}(Fx)\leq (2+o(1))n$.
\end{theorem}

In words, the discrete Gaussian achieves near equality in the uncertainty principle \eqref{eq.main 1}.
Moreover, numerical evidence suggests that $\operatorname{ns}(x)\operatorname{ns}(Fx)=(2+o(1))n$, i.e., the $2$ is optimal for the discrete Gaussian.
Note that this does not depend on whether $n$ is prime or a perfect square.
Recall that a function $f\in C^{\infty}(\mathbb{R})$ is \textbf{Schwarz} if $\sup_{x\in\mathbb{R}}|x^{\alpha} f^{(\beta)}(x)|<\infty$ for every pair of nonnegative integers $\alpha$ and $\beta$.
We use this to quickly prove a well-known lemma that will help us prove Theorem~\ref{thm.neareqmnup}:

\begin{lemma}
\label{lemma.dft of periodized and discretized functions}
Suppose $f\in C^{\infty}(\mathbb{R})$ is Schwarz and construct a discrete function $x\in\ell(\mathbb{Z}_n)$ by periodizing and sampling $f$ as follows:
\begin{equation}
\label{eq.discretization}
x[j]=\sum_{j'\in\mathbb{Z}}f\bigg(\frac{j}{n}+j'\bigg)
\qquad
\forall j\in\mathbb{Z}_n.
\end{equation}
Then the discrete Fourier transform of $x$ is determined by $\hat{f}(\xi):=\int_{-\infty}^\infty f(t)e^{-2\pi i \xi t}dt$:
\begin{equation*}
(Fx)[k]=\sqrt{n}\sum_{k'\in\mathbb{Z}}\hat{f}(k+k'n)
\qquad
\forall k\in\mathbb{Z}_n.
\end{equation*}
\end{lemma}

\begin{proof}
Since $f$ is Schwarz, we may apply the Poisson summation formula:
\[
x[j]
=\sum_{j'\in\mathbb{Z}}f\bigg(\frac{j}{n}+j'\bigg)
=\sum_{l\in\mathbb{Z}}\hat{f}(l)e^{2\pi ijl/n}.
\]
Next, the geometric sum formula gives
\begin{align*}
(Fx)[k]
&=\frac{1}{\sqrt{n}}\sum_{j\in\mathbb{Z}_n}\bigg(\sum_{l\in\mathbb{Z}}\hat{f}(l)e^{2\pi ijl/n}\bigg)e^{-2\pi ijk/n}\\
&=\frac{1}{\sqrt{n}}\sum_{l\in\mathbb{Z}}\hat{f}(l)\sum_{j\in\mathbb{Z}_n}\Big(e^{2\pi i(l-k)/n}\Big)^j
=\sqrt{n}\sum_{\substack{l\in\mathbb{Z}\\l\equiv k\bmod n}}\hat{f}(l).
\end{align*}
The result then follows from a change of variables.
\end{proof}

\begin{proof}[Proof of Theorem~\ref{thm.neareqmnup}]
It is straightforward to verify that the function $f(t)=e^{-n\pi t^2}$ is Schwarz.
Note that defining $x$ according to \eqref{eq.discretization} then produces \eqref{eq.discrete gaussian}.
Considering $\hat{f}(\xi)=n^{-1/2}e^{-\pi \xi^2/n}$, one may use Lemma~\ref{lemma.dft of periodized and discretized functions} to quickly verify that $Fx=x$.
To prove Theorem~\ref{thm.neareqmnup}, it then suffices to show that $\operatorname{ns}(x)\leq(\sqrt{2}+o(1))\sqrt{n}$.
We accomplish this by bounding $\|x\|_2$ and $\|x\|_1$ separately.

To bound $\|x\|_2$, we first expand a square to get
\begin{equation*}
\|x\|_2^2
=\sum_{j\in \mathbb{Z}_n}\bigg(\sum_{j'\in\mathbb{Z}}e^{-n\pi(\frac{j}{n}+j')^2}\bigg)^2
=\sum_{j\in \mathbb{Z}_n}\sum_{j'\in\mathbb{Z}}\sum_{j''\in\mathbb{Z}}e^{-n\pi[(\frac{j}{n}+j')^2+(\frac{j}{n}+j'')^2]}.
\end{equation*}
Since all of the terms in the sum are nonnegative, we may infer a lower bound by discarding the terms for which $j''\neq j'$.
This yields the following:
\begin{equation*}
\|x\|_2^2
\geq \sum_{j\in\mathbb{Z}_n}\sum_{j'\in\mathbb{Z}}e^{-2n\pi(\frac{j}{n}+j')^2}
=\sum_{k\in\mathbb{Z}}e^{-2\pi k^2/n}
\geq\int_{-\infty}^{\infty}e^{-2\pi x^2/n}dx-1
=\sqrt{\frac{n}{2}}-1,
\end{equation*}
where the last inequality follows from an integral comparison.
Next, we bound $\|x\|_1$ using a similar integral comparison:
\[
\|x\|_1
=\sum_{j\in\mathbb{Z}_n}\sum_{j'\in\mathbb{Z}}e^{-n\pi(\frac{j}{n}+j')^2}
=\sum_{k\in\mathbb{Z}}e^{-\pi k^2/n}
\leq \int_{-\infty}^\infty e^{-\pi x^2/n}dx+1
=\sqrt{n}+1.
\]
Overall, we have
\[
\operatorname{ns}(x)
=\frac{\|x\|_1^2}{\|x\|_2^2}
\leq\frac{(\sqrt{n}+1)^2}{\sqrt{n/2}-1}
=(\sqrt{2}+o(1))\sqrt{n}.\qedhere
\]
\end{proof}

\section{Applications}

Having studied the new uncertainty principles in Theorem~\ref{thm:main}, we now take some time to identify certain consequences in various sparse signal processing applications.
In particular, we report consequences in sparse signal demixing, in compressed sensing with partial Fourier operators, and in the fast detection of sparse signals.

\subsection{Sparse signal demixing}
\label{subsection.demixing}

Suppose a signal $x$ is sparse in the Fourier domain and corrupted by noise $\epsilon$ which is sparse in the time domain (such as speckle).
The goal of demixing is to recover the original signal $x$ given the corrupted signal $z=x+\epsilon$; see~\cite{McCoyCDAB:14} for a survey of various related demixing problems.
Provided $Fx$ and $\epsilon$ are sufficiently sparse, it is known that this recovery can be accomplished by solving
\begin{equation}
\label{eq.demixing program}
v^{\star}\quad:=\quad\operatorname{argmin}\quad\|v\|_1\quad\text{subject to}\quad[I~F]v=Fz,
\end{equation}
where, if successful, the solution $v^{\star}$ is the column vector obtained by concatenating $Fx$ and $\epsilon$; see~\cite{SantosaS:86} for an early appearance of this sort of approach.
To some extent, we know how sparse $Fx$ and $\epsilon$ must be for this $\ell_1$ recovery method to succeed.
Coherence-based guarantees in~\cite{DonohoH:01,DonohoE:03,GribonvalN:07} show that it suffices for $v^{\star}$ to be $k$-sparse with $k=O(\sqrt{n})$, while restricted isometry--based guarantees~\cite{Candes:08,BaraniukDDW:08} allow for $k=O(n)$ if $[I~F]$ is replaced with a random matrix.
This disparity is known as the \textbf{square-root bottleneck}~\cite{Tropp:08b}.
In particular, does $[I~F]$ perform similarly to a random matrix, or is the coherence-based sufficient condition on $k$ also necessary?

In the case where $n$ is a perfect square, it is well known that the coherence-based sufficient condition is also necessary.
Indeed, let $K$ denote the subgroup of $\mathbb{Z}_n$ of size $\sqrt{n}$ and suppose $x=1_K$ and $\epsilon=-1_K$.
Then $[Fx;\epsilon]$ is $2\sqrt{n}$-sparse, and yet $z=0$, thereby forcing $v^\star=0$.
On the other hand, if $n$ is prime, then the additive uncertainty principle of Theorem~\ref{thm.tao up} implies that every member of the nullspace of $[I~F]$ has at least $n+1$ nonzero entries, and so $v^\star\neq0$ in this setting.
Still, considering Figure~\ref{figure.discrete gaussian}, one might expect a problem from a stability perspective.
In this section, we use numerical sparsity to show that $\Phi=[I~F]$ cannot break the square-root bottleneck, even if $n$ is prime.
To do this, we will make use of the following theorem:

\begin{theorem}[see~\cite{KashinT:07,CahillM:arxiv14}]
\label{thm.wps}
Denote $\Delta(y):=\operatorname{argmin}\|x\|_1$ subject to $\Phi x=y$.
Then
\begin{equation}
\label{eq.stability}
\|\Delta(\Phi x)-x\|_2\leq\frac{C}{\sqrt{k}}\|x-x_k\|_1\qquad\forall x\in\mathbb{R}^n
\end{equation}
if and only if $\Phi$ satisfies the $(k,c)$-width property.
Furthermore, $C\asymp c$ in both directions of the equivalence.
\end{theorem}

Take $x$ as defined in \eqref{eq.discrete gaussian}.
Then $[x;-x]$ lies in the nullspace of $[I~F]$ and
\[
\operatorname{ns}([x;-x])
=\frac{(2\|x\|_1)^2}{2\|x\|_2^2}
=2\operatorname{ns}(x)
\leq(2\sqrt{2}+o(1))\sqrt{n},
\]
where the last step follows from the proof of Theorem~\ref{thm.neareqmnup}.
As such, $[I~F]$ satisfies the $(k,c)$-width property for some $c$ independent of $n$ only if $k=O(\sqrt{n})$.
Furthermore, Theorem~\ref{thm.wps} implies that stable demixing by $\ell_1$ reconstruction requires $k=O(\sqrt{n})$, thereby proving the necessity of the square-root bottleneck in this case.

It is worth mentioning that the restricted isometry property is a sufficient condition for \eqref{eq.stability} (see~\cite{Candes:08}, for example), and so by Theorem~\ref{thm.iu rip}, one can break the square-root bottleneck by replacing the $F$ in $[I~F]$ with a random unitary matrix.
This gives a uniform demixing guarantee which is similar to those provided by McCoy and Tropp~\cite{McCoyT:14}, though the convex program they consider differs from \eqref{eq.demixing program}.

\subsection{Compressed sensing with partial Fourier operators}

Consider the random $m\times n$ matrix obtained by drawing rows uniformly with replacement from the $n\times n$ discrete Fourier transform matrix.
If $m=\Omega_\delta(k\operatorname{polylog} n)$, then the resulting partial Fourier operator satisfies the restricted isometry property, and this fact has been dubbed the \textbf{uniform uncertainty principle}~\cite{CandesT:06}.
A fundamental problem in compressed sensing is determining the smallest number $m$ of random rows necessary.
To summarize the progress to date, Cand\`{e}s and Tao~\cite{CandesT:06} first found that $m=\Omega_\delta(k\log^6n)$ rows suffice, then Rudelson and Vershynin~\cite{RudelsonV:08} proved $m=\Omega_\delta(k\log^4n)$, and recently, Bourgain~\cite{Bourgain:14} achieved $m=\Omega_\delta(k\log^3n)$; Nelson, Price and Wootters~\cite{NelsonPW:14} also achieved $m=\Omega_\delta(k\log^3n)$, but using a slightly different measurement matrix.
In this subsection, we provide a lower bound: in particular, $m=\Omega_\delta(k\log n)$ is necessary whenever $k$ divides $n$.
Our proof combines ideas from the multiplicative uncertainty principle and the classical problem of coupon collecting.

The coupon collector's problem asks how long it takes to collect all $k$ coupons in an urn if you repeatedly draw one coupon at a time randomly with replacement.
It is a worthwhile exercise to prove that the expected number of trials scales like $k\log k$.
We will require even more information about the distribution of the random number of trials:

\begin{theorem}[see~\cite{ErdosR:61,Csorgo:93}]
\label{thm.coupon collecting}
Let $T_k$ denote the random number of trials it takes to collect $k$ different coupons, where in each trial, a coupon is drawn uniformly from the $k$ coupons with replacement.
\begin{itemize}
\item[(a)]
For each $a\in\mathbb{R}$,
\[
\lim_{k\rightarrow\infty}\operatorname{Pr}\Big(T_k\leq k\log k+ak\Big)
=e^{-e^{-(a+\gamma)}},
\]
where $\gamma\approx0.5772$ denotes the Euler--Mascheroni constant.
\item[(b)]
There exists $c>0$ such that for each $k$,
\[
\sup_{a\in\mathbb{R}}\bigg|\operatorname{Pr}\Big(T_k\leq k\log k+ak\Big)-e^{-e^{-(a+\gamma)}}\bigg|
\leq\frac{c\log k}{k}.
\]
\end{itemize}
\end{theorem}

\begin{lemma}
\label{lemma.parital fourier}
Suppose $k$ divides $n$, and draw $m$ iid rows uniformly from the $n\times n$ discrete Fourier transform matrix to form a random $m\times n$ matrix $\Phi$.
If $m< k\log k$, then the nullspace of $\Phi$ contains a $k$-sparse vector with probability $\geq 0.4-c(\log k)/k$, where $c$ is the constant from Theorem~\ref{thm.coupon collecting}(b).
\end{lemma}

\begin{proof}
Let $K$ denote the subgroup of $\mathbb{Z}_n$ of size $k$, and let $1_K$ denote its indicator function.
We claim that some modulation of $1_K$ resides in the nullspace of $\Phi$ with the probability reported in the lemma statement.
Let $H$ denote the subgroup of $\mathbb{Z}_n$ of size $n/k$.
Then the Fourier transform of each modulation of $1_K$ is supported on some coset of $H$.
Letting $M$ denote the random row indices that are drawn uniformly from $\mathbb{Z}_n$, a modulation of $1_K$ resides in the nullspace of $\Phi$ precisely when $M$ fails to intersect the corresponding coset of $H$.
As there are $k$ cosets, each with probability $1/k$, this amounts to a coupon-collecting problem (explicitly, each ``coupon'' is a coset, and we ``collect'' the cosets that $M$ intersects).
The result then follows immediately from Theorem~\ref{thm.coupon collecting}(b):
\[
\operatorname{Pr}(T_k\leq m)
\leq e^{-e^{-(m/k-\log k+\gamma)}}+\frac{c\log k}{k}
\leq e^{-e^{-\gamma}}+\frac{c\log k}{k}
\leq 0.6+\frac{c\log k}{k}.\qedhere
\]
\end{proof}

Presumably, one may remove the divisibility hypothesis in Lemma~\ref{lemma.parital fourier} at the price of weakening the conclusion.
We suspect that the new conclusion would declare the existence of a vector $x$ of numerical sparsity $k$ such that $\|\Phi x\|_2\ll\|x\|_2$.
If so, then $\Phi$ fails to satisfy the so-called \textbf{robust width property}, which is necessary and sufficient for stable and robust reconstruction by $\ell_1$ minimization~\cite{CahillM:arxiv14}.
For the sake of simplicity, we decided not to approach this, but we suspect that modulations of the discrete Gaussian would adequately fill the role of the current proof's modulated indicator functions.

What follows is the main result of this subsection:

\begin{theorem}
\label{theorem.parital fourier RIP}
Let $k$ be sufficiently large, suppose $k$ divides $n$, and draw $m$ iid rows uniformly from the $n\times n$ discrete Fourier transform matrix to form a random $m\times n$ matrix $\Phi$.
Take $\delta<1/3$.
Then $\Phi$ satisfies the $(k,\delta)$-restricted isometry property with probability $\geq2/3$ only if
\[
m\geq C(\delta)k\log(en),
\]
where $C(\delta)$ is some constant depending only on $\delta$.
\end{theorem}

\begin{proof}
In the event that $\Phi$ satisfies $(k,\delta)$-RIP, we know that no $k$-sparse vector lies in the nullspace of $\Phi$.
Therefore, Lemma~\ref{lemma.parital fourier} implies
\begin{equation}
\label{eq.coupon bound}
m\geq k\log k, 
\end{equation}
since otherwise $\Phi$ fails to be $(k,\delta)$-RIP with probability $\geq0.4-c(\log k)/k>1/3$, where the last step uses the fact that $k$ is sufficiently large.
Next, we leverage standard techniques from compressed sensing:
$(k,\delta)$-RIP implies \eqref{eq.stability} with $C=C_1(\delta)$ (see Theorem~3.3 in~\cite{CaiZ:13}), which in turn implies
\begin{equation}
\label{eq.gelfand bound}
m\geq C_2(\delta)k\log\bigg(\frac{en}{k}\bigg)
\end{equation}
by Theorem~11.7 in~\cite{FoucartR:13}.
Since $\Phi$ is $(k,\delta)$-RIP with positive probability, we know there exists an $m\times n$ matrix which is $(k,\delta)$-RIP, and so $m$ must satisfy \eqref{eq.gelfand bound}.
Combining with \eqref{eq.coupon bound} then gives
\begin{align*}
m
&\geq\max\bigg\{k\log k,C_2(\delta)k\log\bigg(\frac{en}{k}\bigg)\bigg\}.
\end{align*}
The result then follows from applying the bound $\max\{a,b\}\geq(a+b)/2$ and then taking $C(\delta):=(1/2)\min\{1,C_2(\delta)\}$.
\end{proof}

We note that the necessity of $k\log n$ random measurements contrasts with the proportional-growth asymptotic adopted in~\cite{BlanchardCT:11} to study the restricted isometry property of Gaussian matrices.
Indeed, it is common in compressed sensing to consider phase transitions in which $k$, $m$ and $n$ are taken to infinity with fixed ratios $k/m$ and $m/n$.
However, since random partial Fourier operators fail to be restricted isometries unless $m=\Omega_\delta(k\log n)$, such a proportional-growth asymptotic fails to capture the so-called \textit{strong phase transition} of these operators~\cite{BlanchardCT:11}.

The proof of Theorem~\ref{theorem.parital fourier RIP} relies on the fact that the measurements are drawn at random.
By contrast, it is known that \textit{every} $m\times n$ partial Hadamard operator fails to satisfy $(k,\delta)$-RIP unless $m=\Omega_\delta(k\log n)$~\cite{Talagrand:98,GuedonMPT:08}.
We leave the corresponding deterministic result in the Fourier case for future work.

\subsection{Fast detection of sparse signals}

The previous subsection established fundamental limits on the number of Fourier measurements necessary to perform compressed sensing with a uniform guarantee.
However, for some applications, signal reconstruction is unnecessary.
In this subsection, we consider one such application, namely sparse signal detection, in which the goal is to test the following hypotheses:
\begin{align*}
H_0&:~x=0\\
H_1&:~\|x\|_2^2=\frac{n}{k},~\|x\|_0\leq k.
\end{align*}
Here, we assume we know the 2-norm of the sparse vector we intend to detect, and we set it to be $\sqrt{n/k}$ without loss of generality (this choice of scaling will help us interpret our results later).
We will assume the data is accessed according to the following query--response model:

\begin{definition}[Query--response model]\label{def.model}
If the $i$th query is $j_i\in\mathbb{Z}_n$, then the $i$th response is $(Fx)[j_i]+\epsilon_i$,
where the $\epsilon_i$'s are iid complex random variables with some distribution such that
\[
\mathbb{E}|\epsilon_i|=\alpha,
\qquad
\mathbb{E}|\epsilon_i|^2=\beta^2.
\]
\end{definition}

The coefficient of variation $v$ of $|\epsilon_i|$ is defined as
\begin{equation}
\label{eq.coef var}
v
=\frac{\sqrt{\operatorname{Var}|\epsilon_i|}}{\mathbb{E}|\epsilon_i|}
=\frac{\sqrt{\beta^2-\alpha^2}}{\alpha}.
\end{equation}
Note that for any scalar $c\neq0$, the mean and variance of $|c\epsilon_i|$ are $|c|\alpha$ and $|c|^2\operatorname{Var}|\epsilon_i|$, respectively.
As such, $v$ is scale invariant and is simply a quantification of the ``shape'' of the distribution of $|\epsilon_i|$.
We will evaluate the responses to our queries with  an $\ell_1$ detector, defined below.

\begin{definition}[$\ell_1$ detector]
Fix a threshold $\tau$.
Given responses $\{y_i\}_{i=1}^m$ from the query--response model, if
\[
\sum_{i=1}^m|y_i|>\tau,
\]
then reject $H_0$.
\end{definition}

The following is the main result of this section:

\begin{theorem}
\label{thm.searchsparse}
Suppose $\alpha\leq1/(8k)$.
Randomly draw $m$ indices uniformly from $\mathbb{Z}_n$ with replacement, input them into the query--response model and apply the $\ell_1$ detector with threshold $\tau=2m\alpha$ to the responses.
Then
\begin{equation}
\label{eq.false positive}
\operatorname{Pr}\bigg(\text{reject }H_0~\bigg|~H_0\bigg)\leq p
\end{equation}
and
\begin{equation}
\label{eq.false negative}
\operatorname{Pr}\bigg(\text{fail to reject }H_0~\bigg|~H_1\bigg)\leq p
\end{equation}
provided $m\geq(8k+2v^2)/p$, where $v$ is the coefficient of variation defined in \eqref{eq.coef var}.
\end{theorem}

In words, the probability that the $\ell_1$ detector delivers a false positive is at most $p$, as is the probability that it delivers a false negative.
These error probabilities can be estimated better given more information about the distribution of the random noise, and presumably, the threshold $\tau$ can be modified to decrease one error probability at the price of increasing the other.
Notice that we only use $O(k)$ samples in the Fourier domain to detect a $k$-sparse signal.
Since the sampled indices are random, it will take $O(\log n)$ bits to communicate each query, leading to a total computational burden of $O(k\log n)$ operations.
This contrasts with the state-of-the-art sparse fast Fourier transform algorithms which require $\Omega(k\log(n/k))$ samples and take $O(k\operatorname{polylog}n)$ time (see~\cite{IndykK:14} and references therein).
We suspect $k$-sparse signals cannot be detected with substantially fewer samples (in the Fourier domain or any domain).

We also note that the acceptable noise magnitude $\alpha=O(1/k)$ is optimal in some sense.
To see this, consider the case where $k$ divides $n$ and $x$ is a properly scaled indicator function of the subgroup of size $k$.
Then $Fx$ is the indicator function of the subgroup of size $n/k$.
(Thanks to our choice of scaling, each nonzero entry in the Fourier domain has unit magnitude.)
Since a proportion of $1/k$ entries is nonzero in the Fourier domain, we can expect to require $O(k)$ random samples in order to observe a nonzero entry, and the $\ell_1$ detector will not distinguish the entry from accumulated noise unless $\alpha=O(1/k)$.

Before proving Theorem~\ref{thm.searchsparse}, we first prove a couple of lemmas. 
We start by estimating the probability of a false positive:

\begin{lemma}
\label{lem.falsepos}
Take $\epsilon_1,\ldots,\epsilon_m$ to be iid complex random variables with $\mathbb{E}|\epsilon_i|=\alpha$ and $\mathbb{E}|\epsilon_i|^2=\beta^2$.
Then
\[
\operatorname{Pr}\bigg(\sum_{i=1}^m|\epsilon_i|> 2m\alpha\bigg)\leq p
\]
provided $m\geq v^2/p$, where $v$ is the coefficient of variation of $|\epsilon_i|$ defined in \eqref{eq.coef var}.
\end{lemma}

\begin{proof}
Denoting $X:=\sum_{i=1}^m|\epsilon_i|$, we have $\mathbb{E}X=m\alpha$ and $\operatorname{Var}X=m(\beta^2-\alpha^2)$.
Chebyshev's inequality then gives
\[
\operatorname{Pr}\bigg(\sum_{i=1}^m|\epsilon_i|-m\alpha>t\bigg)
\leq\operatorname{Pr}(|X-\mathbb{E}X|>t)
\leq\frac{\operatorname{Var}X}{t^2}
=\frac{m(\beta^2-\alpha^2)}{t^2}.
\]
Finally, we take $t=m\alpha$ to get
\[
\operatorname{Pr}\bigg(\sum_{i=1}^m|\epsilon_i|>2m\alpha\bigg)
\leq m\frac{(\beta^2-\alpha^2)}{(m\alpha)^2}
=\frac{\beta^2-\alpha^2}{m\alpha^2}
\leq\frac{\beta^2-\alpha^2}{\alpha^2}\cdot\frac{p}{v^2}
=p.\qedhere
\]
\end{proof}

Next, we leverage the multiplicative uncertainty principle in Theorem~\ref{thm:main} to estimate moments of noiseless responses:

\begin{lemma}
\label{lem.1}
Suppose $\|x\|_0\leq k$ and $\|x\|_2^2=n/k$. Draw $j$ uniformly from $\mathbb{Z}_n$ and define $Y:=|(Fx)[j]|$. Then
\[
\mathbb{E}Y\geq\frac{1}{k},
\qquad
\mathbb{E}Y^2=\frac{1}{k}.
\]
\end{lemma}

\begin{proof}
Recall that $\operatorname{ns}(x)\leq\|x\|_0\leq k$.
With this, Theorem~\ref{thm:main} gives
\[
n
\leq\operatorname{ns}(x)\operatorname{ns}(Fx)
\leq k\operatorname{ns}(Fx).
\]
We rearrange and apply the definition of numerical sparsity to get
\[
\frac{n}{k}
\leq\operatorname{ns}(Fx)
=\frac{\|Fx\|_1^2}{\|Fx\|_2^2}
=\frac{\|Fx\|_1^2}{\|x\|_2^2}
=\frac{\|Fx\|_1^2}{n/k},
\]
where the second to last equality is due to Parseval's identity.
Thus, $\|Fx\|_1\geq n/k$.
Finally,
\[
\mathbb{E}Y
=\frac{1}{n}\sum_{j\in\mathbb{Z}_n}|(Fx)[j]|
=\frac{1}{n}\|Fx\|_1
\geq\frac{1}{k}
\]
and
\[
\mathbb{E}Y^2
=\frac{1}{n}\sum_{j\in\mathbb{Z}_n}|(Fx)[j]|^2
=\frac{1}{n}\|Fx\|_2^2
=\frac{1}{k}.\qedhere
\]
\end{proof}

\begin{proof}[Proof of Theorem~\ref{thm.searchsparse}]
Lemma~\ref{lem.falsepos} gives \eqref{eq.false positive}, and so it remains to prove \eqref{eq.false negative}.
Denoting $Y_i:=|(Fx)[j_i]|$, we know that $|y_i|\geq Y_i-|\epsilon_i|$, and so
\begin{equation}
\label{eq.parta}
\operatorname{Pr}\bigg(\sum_{i=1}^m|y_i|\leq2ma\bigg)
\leq\operatorname{Pr}\bigg(\sum_{i=1}^mY_i-\sum_{i=1}^m|\epsilon_i|\leq2ma\bigg).
\end{equation}
For notational convenience, put $Z:=\sum_{i=1}^mY_i-\sum_{i=1}^m|\epsilon_i|$.
We condition on the size of the noise and apply Lemma~\ref{lem.falsepos} with the fact that $m\geq v^2/(p/2)$ to bound \eqref{eq.parta}:
\begin{align}
\label{eq.partb}
\operatorname{Pr}(Z\leq2m\alpha)
&=\operatorname{Pr}\bigg(Z\leq2m\alpha~\bigg|~\sum_{i=1}^m|\epsilon_i|>2m\alpha\bigg)\operatorname{Pr}\bigg(\sum_{i=1}^m|\epsilon_i|>2m\alpha\bigg)\notag\\
&\qquad+\operatorname{Pr}\bigg(Z\leq2m\alpha~\bigg|~\sum_{i=1}^m|\epsilon_i|\leq2m\alpha\bigg)\operatorname{Pr}\bigg(\sum_{i=1}^m|\epsilon_i|\leq2m\alpha\bigg)\notag\\
&\leq \frac{p}{2}+\operatorname{Pr}\bigg(\sum_{i=1}^mY_i\leq4m\alpha\bigg).
\end{align}
Now we seek to bound the second term of \eqref{eq.partb}.
Taking $X=\sum_{i=1}^mY_i$, Lemma~\ref{lem.1} gives $\mathbb{E}X\geq m/k$ and $\operatorname{Var}X=m\operatorname{Var}Y_i\leq m\mathbb{E}Y_i^2=m/k$.
As such, applying Chebyshev's inequality gives
\[
\operatorname{Pr}\bigg(\sum_{i=1}^mY_i<\frac{m}{k}-t\bigg)
\leq \operatorname{Pr}(X\leq\mathbb{E}X-t)
\leq \operatorname{Pr}(|X-\mathbb{E}X|>t)
\leq\frac{\operatorname{Var}(X)}{t^2}
\leq\frac{m}{kt^2}.
\]
Recalling that $\alpha\leq1/(8k)$, we take $t=m/(2k)$ to get
\begin{equation}
\label{eq.partc}
\operatorname{Pr}\bigg(\sum_{i=1}^mY_i\leq4m\alpha\bigg)
\leq \operatorname{Pr}\bigg(\sum_{i=1}^mY_i\leq\frac{m}{2k}\bigg)
=\operatorname{Pr}\bigg(\sum_{i=1}^mY_i\leq\frac{m}{k}-t\bigg)
\leq\frac{m}{kt^2}
=\frac{4k}{m}
\leq \frac{p}{2},
\end{equation}
where the last step uses the fact that $m\geq8k/p$.
Combining \eqref{eq.parta}, \eqref{eq.partb}, and \eqref{eq.partc} gives the result.
\end{proof}

\section*{Acknowledgments}

The authors thank Laurent Duval, Joel Tropp, and the anonymous referees for multiple suggestions that significantly improved the presentation of our results and our discussion of the relevant literature.
ASB was supported by AFOSR Grant No.\ FA9550-12-1-0317.
DGM was supported by an AFOSR Young Investigator Research Program award, NSF Grant No.\ DMS-1321779, and AFOSR Grant No.\ F4FGA05076J002.
The views expressed in this article are those of the authors and do not reflect the official policy or position
of the United States Air Force, Department of Defense, or the U.S.\ Government.

\end{document}